\documentclass[letterpaper]{article} % DO NOT CHANGE THIS
\usepackage{aaai2026}  % DO NOT CHANGE THIS
\usepackage{times}  % DO NOT CHANGE THIS
\usepackage{helvet}  % DO NOT CHANGE THIS
\usepackage{courier}  % DO NOT CHANGE THIS
\usepackage[hyphens]{url}  % DO NOT CHANGE THIS
\usepackage{graphicx} % DO NOT CHANGE THIS
\usepackage{amsthm}
\usepackage{amsmath}
\usepackage{amssymb}
\usepackage{natbib}  % DO NOT CHANGE THIS AND DO NOT ADD ANY OPTIONS TO IT
\usepackage{caption} % DO NOT CHANGE THIS AND DO NOT ADD ANY OPTIONS TO IT
\usepackage{algorithm}
\usepackage{algorithmic}

\usepackage{newfloat}
\usepackage{listings}
\DeclareCaptionStyle{ruled}{labelfont=normalfont,labelsep=colon,strut=off} % DO NOT CHANGE THIS
\floatstyle{ruled}
\newfloat{listing}{tb}{lst}{}
\floatname{listing}{Listing}
\newtheorem{theorem}{Theorem}
\newtheorem{lemma}{Lemma}
\newcommand{\highlight}[1]{\textbf{\underline{#1}}}
\title{TADT-CSA: Temporal Advantage Decision Transformer with Contrastive State Abstraction for Generative Recommendation}
\author{
    Xiang Gao\textsuperscript{\rm 1},
    Tianyuan Liu\textsuperscript{\rm 2}\thanks{This work was done when Tianyuan Liu was an intern in Kuaishou Technology.},
    Yisha Li\textsuperscript{\rm 1},
    Jingxin Liu\textsuperscript{\rm 1}\thanks{Corresponding Author.},
    Lexi Gao\textsuperscript{\rm 1},
    Xin Li\textsuperscript{\rm 1},
    Haiyang Lu\textsuperscript{\rm 1},
    Liyin Hong\textsuperscript{\rm 1}
}
\affiliations{
    \textsuperscript{\rm 1}Kuaishou Technology, Beijing, China \\
    \textsuperscript{\rm 2}School of Artificial Intelligence, Nanjing University, China \\
    \{gaoxiang12, liyisha, liujingxin05, gaolexi, lixin05, luhaiyang, hongliyin\}@kuaishou.com, liutianyuan@lamda.nju.edu.cn
}

\usepackage{bibentry}
\usepackage{booktabs}
\usepackage{multirow}
\usepackage{array}
\usepackage{siunitx}
\usepackage{makecell}

\newcommand{\best}[1]{\textbf{#1}}

\begin{document}

\maketitle

\begin{abstract}
With the rapid advancement of Transformer-based Large Language Models (LLMs), generative recommendation has shown great potential in enhancing both the accuracy and semantic understanding of modern recommender systems. Compared to LLMs, the Decision Transformer (DT) is a lightweight generative model applied to sequential recommendation tasks. However, DT faces challenges in trajectory stitching, often producing suboptimal trajectories. Moreover, due to the high dimensionality of user states and the vast state space inherent in recommendation scenarios, DT can incur significant computational costs and struggle to learn effective state representations. To overcome these issues, we propose a novel \textbf{T}emporal \textbf{A}dvantage \textbf{D}ecision \textbf{T}ransformer with \textbf{C}ontrastive \textbf{S}tate \textbf{A}bstraction (TADT-CSA) model. Specifically, we combine the conventional Return-To-Go (RTG) signal with a novel temporal advantage (TA) signal that encourages the model to capture both long-term returns and their sequential trend. Furthermore, we integrate a contrastive state abstraction module into the DT framework to learn more effective and expressive state representations. Within this module, we introduce a TA–conditioned State Vector Quantization (TAC-SVQ) strategy, where the TA score guides the state codebooks to incorporate contextual token information. Additionally, a reward prediction network and a contrastive transition prediction (CTP) network are employed to ensure that the state codebook preserves both the reward information of the current state and the transition information between adjacent states. Empirical results on both public datasets and an online recommendation system demonstrate the effectiveness of the TADT-CSA model and its superiority over baseline methods.

\end{abstract}

% Uncomment the following to link to your code, datasets, an extended version or similar.
% You must keep this block between (not within) the abstract and the main body of the paper.
% \begin{links}
%     \link{Code}{https://aaai.org/example/code}
%     \link{Datasets}{https://aaai.org/example/datasets}
%     \link{Extended version}{https://aaai.org/example/extended-version}
% \end{links}

\section{Introduction}
Currently, most reinforcement learning (RL) applications in industrial recommender systems (RS) formulate the problem as an infinite-horizon, request-level Markov Decision Process (MDP) \cite{cai2023reinforcing, liu2024supervised, zhang2024unex}, and optimize long-term rewards through Temporal Difference (TD) error-based bootstrapping. However, in real-world RS environments, data distributions often exhibit significant fluctuations, especially during transitions from traffic high peaks to low peaks. The collected $(s_t, a_t, r_t, s_{t+1})$ samples can be highly noisy and stochastic, making it difficult for RL agents to obtain accurate Q-value estimates. Moreover, TD-based methods heavily rely on local information, which limits their ability to capture the long-term evolution of user interests and value accumulation over time.

The Decision Transformer (DT) \cite{chen2021decision} has emerged as a promising sequence modeling approach for offline reinforcement learning. Unlike traditional RL frameworks, DT treats policy learning as a conditional generation task, enabling it to model long-range dependencies more effectively. However, DT lacks the ability to perform trajectory stitching, a crucial capability that allows offline RL models to learn optimal policies from sub-optimal trajectories \cite{yamagata2023q}. As a result, DT may fail to learn optimal policies when high-return trajectories are rare or the environment is stochastic \cite{brandfonbrener2022does}.

Recently, DT has been adopted in generative recommendation systems \cite{wang2023generative} and long-sequence decision-making tasks. Compared with large language models (LLMs) \cite{ji2024genrec}, DT offers a lightweight generative architecture, making it suitable for deployment in industrial RS, such as auto bidding scenario \cite{gao2025generative}, with relatively low computational overhead.

Despite its advantages, existing DT-based recommendation models \cite{zhao2023user, wang2023causal, wang2024retentive, chen2024maximum, liu2024sequential} typically rely on simple embedding layers or shallow encoders for state representation learning. In industrial RS, however, the user base often reaches tens of millions, while the available items scales to billions. Consequently, the state space becomes the Cartesian product of these high-dimensional feature spaces, resulting in an extremely large and sparse representation space. Learning effective state representations under such conditions is highly challenging. 

To address the above issues, we propose a novel \textbf{T}emporal \textbf{A}dvantage \textbf{D}ecision \textbf{T}ransformer with \textbf{C}ontrastive \textbf{S}tate \textbf{A}bstraction (TADT-CSA) model.
We integrate the conventional Return-to-Go (RTG) score with a novel Temporal Advantage (TA) score to jointly condition the policy on both long-term returns and recent temporal trends. Secondly, we propose a pairwise ranking loss for the TADT to alleviate the limitations of the original behavior cloning (BC) objective and encourage the policy to improve along directions with high RTG scores. Furthermore, we propose a novel Contrastive State Abstraction (CSA) module for more accurate and efficient state representation learning. Within this module, the TA-conditioned State Vector Quantization (TAC-SVQ) strategy and state auxiliary networks help preserve useful trajectory and MDP information during state representation learning.

In summary, our contributions are as follows:
\begin{itemize}
    \item We introduce a novel TADT-CSA model for generative recommendation, especially tailored for large-scale industrial RS scenarios where the online environment is highly noisy and stochastic.
    \item We incorporate the temporal advantage score into the TADT model and propose a pairwise ranking loss to prevent the model from falling into simple behavior cloning.
    \item We propose a Contrastive State Abstraction (CSA) module for effective state representation learning, which is particularly beneficial when the state space is extremely large and sparse.
    \item Through extensive offline evaluation, online simulation and online A/B Tests, we demonstrate that the proposed TADT-CSA model outperforms existing methods and shows strong practical value in industrial applications.
\end{itemize}

\begin{figure*}[htbp]
\centering
\includegraphics[width=0.8\textwidth]{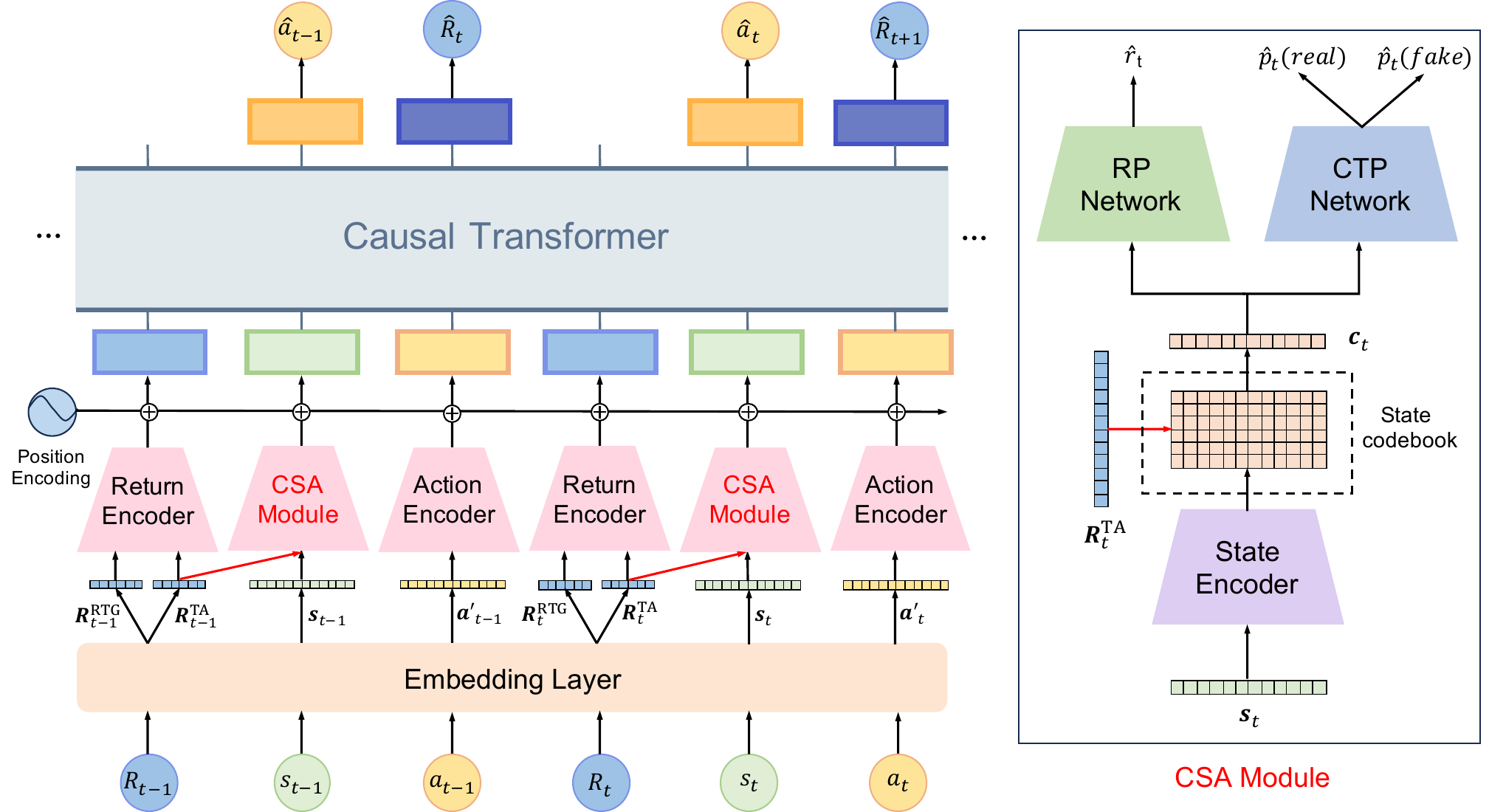}
\caption{Overall framework of the TADT-CSA model. $R$ denotes the return-conditioned signal shown in Eq \ref{eq:4}, RP represents reward prediction, and CTP stands for contrastive transition prediction.}
\label{fig:1}
\end{figure*}

\section{Problem Formulation}
Traditional recommendation systems typically adopt a discriminative paradigm, focusing on user-item scoring and ranking. In contrast, generative sequential recommendation frames the task as a sequence generation problem, where the model learns to autoregressively generate future user interactions based on historical trajectories and target objectives.

From the RL perspective, the sequential recommendation problem can be modeled using a MDP. Formally, let $\mathcal{U} = \left\{u_1, \cdots, u_n\right\}$ denote the set of users, and $\mathcal{I} = \left\{i_1, \cdots, i_m\right\}$ denote the candidate item set. The MDP for sequential recommendation can then be represented by the tuple $\left(\mathcal{S}, \mathcal{A}, \mathcal{P}, \mathcal{R}, \gamma\right)$, defined as follows:
\begin{itemize}
\item State Space $\mathcal{S}$: The state $s_t \in \mathcal{S}$ represents a user's current situation, often composed of static user features and historical interaction information. 
\item Action Space $\mathcal{A}$: The action $a_t \in \mathcal{A}$ corresponds to the recommended item list at time step $t$. 
\item Transition Probability $\mathcal{P}$: The transition function $p(s_{t+1} | s_t, a_t) \in \mathcal{P}$ defines the probability of transitioning to $s_{t+1}$ after taking action $a_t$ in state $s_t$. 
\item Reward $\mathcal{R}$: The reward function $\mathcal{R}$ maps state-action pairs to real-valued rewards. In practice, $r_t$ is often derived from immediate user feedback signals such as Click-Through Rate (CTR), or long-term engagement metrics like Return Frequency.
\item Discount Factor $\gamma$: $\gamma \in [0, 1]$ is a scalar that balances the trade-off between immediate and future rewards.
\end{itemize}

Based on this formulation, user-item interactions are represented as trajectories:
\begin{equation}
\tau = (s_1, a_1, r_1, \cdots, r_t, s_t, a_t, \cdots, s_T, a_T, r_T)    
\end{equation}
where $T$ denotes the maximum length of the trajectory. These trajectories can be collected into an offline dataset $\mathcal{D} = \left\{\tau_{u}\right\}$ , which serves as the training data for offline RL methods aiming to learn optimal recommendation policies. In the inference stage, given an initial state and a target return signal, a DT-based model autoregressively generates the sequence of recommended items, enabling personalized and temporally coherent recommendations.

\section{Method}
In this section, we present the proposed TADT-CSA model. We first introduce the overall framework, as illustrated in Fig. \ref{fig:1}, followed by a detailed description of the TADT architecture. Subsequently, we elaborate on the CSA module, including its design principles and theoretical analysis.
\subsection{Temporal Advantage Decision Transformer}
Due to the high level of noise and significant fluctuations in the data distribution within RS, TD error-based RL models struggle to produce accurate Q-value estimates and learn optimal policies. To address these challenges, we propose a TADT model based on long sequence modeling.

\subsubsection{Temporal Advantage Score} We introduce a novel Temporal Advantage (TA) score to enhance the standard Return-to-Go (RTG) signal with temporal trend information. Unlike RTG, which only captures cumulative future rewards, the TA score not only reflects the long-term reward but also preserves temporal trend information, enabling the model to better understand how recent decisions influence long-term objectives. Formally, given a trajectory $\tau = (s_1, a_1, r_1, \cdots, s_T, a_T, r_T)$, the return-to-go $R^{\text{RTG}}_t$ is defined as a discounted sum of future rewards:
\begin{equation}
R^{\text{RTG}}_t = \sum_{i = t}^{T} \gamma^{i - t} r_i
\end{equation}
Based on this, we define the TA score $R^{\text{TA}}_t$ as:
\begin{equation}
R^{\text{TA}}_t = \sum_{i = 2}^{t} \gamma^{t - i}\left(R^{\text{RTG}}_i - R^{\text{RTG}}_{i - 1}\right)
\end{equation}
which essentially computes a discounted accumulation of RTG differences over time. The initial TA score $R^{\text{TA}}_0$ is set to zero due to the lack of previous RTG values.

The new return-conditioned signal $R_t$ , used for policy generation, is then defined as a joint representation of both components:
\begin{equation}
R_t = [R^{\text{RTG}}_t, R^{\text{TA}}_t]
\label{eq:4}
\end{equation}

Compared with RTG, the TA score captures richer temporal dynamics by reflecting how recent actions influence long-term outcomes, even when RTG values are similar. This enables stronger gradient signals during training. For example, two sequences with identical RTG may exhibit distinct TA patterns, one showing steady growth, the other a rapid surge, providing more discriminative power for decision-making. This property makes the TA score especially valuable in large-scale RS, where RTG values can be sparse or poorly differentiated, and subtle changes in user behavior carry meaningful signals.

\subsubsection{TADT architecture}As illustrated in Fig \ref{fig:1}, a trajectory $\tau$ of length $T$ is decomposed into a sequence $(R_t, s_t, a_t)$ with total length $3T$, which is then fed into the Embedding Layer to obtain the corresponding state embedding $\mathbf{s}_t$, action embedding $\mathbf{a'}_t$ and return embedding $\mathbf{R}'_t = [\mathbf{R}^{\text{RTG}}_t, \mathbf{R}^{\text{TA}}_t]$. Then through the return encoder and action encoder, which are commonly modeled by Multi-Layer Perceptrons(MLPs), we can obtain the return token embedding $\mathbf{R}_t$ and action token embedding $\mathbf{a}_t$. Given the high dimensionality and large magnitude of the state space, particularly in large-scale RS, we introduce the Contrastive State Abstraction (CSA) module to compress the high-dimensional state embeddings $\mathbf{s}_t$ into low-dimensional vector-quantized embeddings $\mathbf{c}_t$ with codebook id $c_t$.

Finally, the trajactory token representation list, i.e. $\left(\mathbf{R}_1, \mathbf{c}_1, \mathbf{a}_1 \cdots, \mathbf{R}_T, \mathbf{c}_T, \mathbf{a}_T\right)$ are fed into a Causal Transformer model \cite{vaswani2017attention}, such as GPT \cite{radford2018improving}. The predicted action probability $p(a_t |\tau_{1:t-1}, s_t, R_t)$ is formulated as follows:
\begin{equation}
\begin{aligned}
&\hat{\mathbf{a}}_t = \text{GPT}(\mathbf{R}_{1:t-1}, \mathbf{c}_{1:t-1}, \mathbf{a}_{1:t-1}, \mathbf{R}_t, \mathbf{c}_t) \\
&\ell_i = \hat{\mathbf{a}}_t^\top \mathbf{\Phi}(i), \\
&\hat{p}(a_t \mid \tau_{1:t-1}, s_t, R_t) = \text{softmax}_{\mathcal{I}}(\ell_i) \\
\end{aligned}
\end{equation}
where $\mathcal{I}$ denotes the candidate item set, and $\mathbf{\Phi}(\cdot)$ is a learnable item embedding layer that generates item embeddings.

Additionally, the predicted return $\hat{R}_t = [\hat{R}^{\text{RTG}}_{t}, \hat{R}^{\text{TA}}_{t}]$ can be derived as follows:
\begin{equation}
\begin{aligned}
&\hat{\mathbf{R}}_t = \text{GPT}(\mathbf{R}_{1:t-1}, \mathbf{c}_{1:t-1}, \mathbf{a}_{1:t-1}) \\
& \hat{R}_t = \text{MLP}(\hat{\mathbf{R}}_t) \\
\end{aligned}
\end{equation}
Notably, the predicted return $\hat{R}_t$ can be used to generate a new return token during the inference phase, enabling dynamic adaptation based on historical interactions.

Since DT often struggles with stitching together optimal trajectories from sub-optimal ones, it may perform poorly in stochastic environments and heavily rely on high-return trajectories to achieve satisfactory performance. To address this limitation, we propose a novel pairwise ranking loss that complements the BC objective within the DT framework. 

For each time step $t$, we utilize the real RTG score to explicitly guide the policy, encouraging actions with higher RTG scores while discouraging those with lower ones. Notably, the policy distribution is conditioned on the sub-trajectory $(s_1, a_1, R_1, \cdots, s_t)$, which necessitates trajectory matching for constructing pairwise training samples. However, in large-scale RS, where states are typically high-dimensional, directly matching sub-trajectories based on $(s_1, a_1, R_1, \cdots, s_t)$ is computationally expensive.

To reduce the computational cost of sub-trajectory
matching, we instead use the quantized codebook indices $c_t$ as compact representations of the original states $s_t$. Specifically, we apply a sequence hashing technique based on the transformed sub-trajectory $(c_1, a_1, R_1, \cdots, c_t)$, which significantly improves efficiency. For each batch, we hash these sub-trajectories into a set of groups:
\begin{equation}
G_t = \left\{g_{t,1}, \cdots, g_{t,K}\right\}
\end{equation}
where $K$ denotes the number of groups.

We then introduce a quantile-based pairwise ranking loss, designed to enhance policy learning within each group. For each group $g_{t,k}$, we compute the $\beta$-quantile of the RTG values and divide the samples into positive and negative subsets. The threshold is defined as $\text{threshold}_{\beta} = \text{Quantile}(R^{\text{RTG}}(g_{t,k}),\beta)$. To further improve computational efficiency, we adopt a top-K selection strategy in place of computing explicit quantile thresholds. Specifically, for each group $g_{t,k}$ of size $N_{g}$, we adopt a QuickSelect algorithm, which is a quicksort-like partition-based method, to efficiently select the top $\lfloor \beta N_{g}\rfloor$ samples with the highest RTG values and treat them as positive examples. This avoids full sorting and reduces the time complexity from $O(n\log n)$ to $O(n)$.

The proposed pairwise ranking loss is defined as:
\begin{equation}
\begin{aligned}
\mathcal{L}_{\text{rank}} &= \sum_{t = 1}^T \sum_{g_{t,k} \in G_t} \sum_{\substack{i \in g_{t,k}^P \\ j \in g_{t,k}^N}}  \mathcal{L}_{\text{pair}}(i, j) \\
\mathcal{L}_{\text{pair}}(i, j) &= -\log \left(\sigma \left(\ell_{a_{t}^{(i)}} - \ell_{a_{t}^{(j)}} - \delta\right)\right)
\end{aligned}
\end{equation}
where $g_{t,k}^P$ and $ g_{t,k}^N$ denote the positive and negative example sets within group $g_{t,k}$, $\sigma$ is the sigmoid function,  $\delta$ is the margin hyperparameter, $\ell_{a_t^{(i)}}$ and $\ell_{a_t^{(j)}}$ are the predicted action logits for samples $i$ and $j$, respectively.

From the perspective of traditional RL, the proposed pairwise loss can be interpreted as an implicit policy improvement mechanism, conceptually similar to the policy gradient loss. This enables the TADT model to go beyond simple behavior cloning by distinguishing high-return actions. From the viewpoint of RS, the loss corresponds to a standard Bayesian Personalized Ranking (BPR) loss, which enhances the model’s ability to rank items effectively according to user preferences. Furthermore, in the context of preference learning \cite{christiano2017deep}, the pairwise loss shares conceptual similarities with the Bradley-Terry model \cite{bradley1952rank}. It leverages the RTG difference between trajectories as an implicit preference signal, guiding the model to better align its predictions with the underlying preference structure induced by RTG.

Moreover, TADT is designed to jointly predict the action $a_t$ and return $R_t$. Accordingly, we define the action prediction loss $\mathcal{L}_{a}$ and the return regression loss $\mathcal{L}_{R}$ as follows:
\begin{equation}
\begin{aligned}
\mathcal{L}_{a} &= - \sum_{i=1}^{N}\sum_{t = 1}^T \log \hat{p}(a^i_t \mid \tau_{1:t-1}, s^i_t, R^i_t) \\
\mathcal{L}_{R} &= \sum_{i=1}^{N}\sum_{t = 1}^T \|\hat{R}^i_t - R^i_t\|^2
\end{aligned}
\end{equation}
, where $a_t^i$, $s_t^i$, and $R_t^i$ denote the ground-truth action, state, and return of the $i$-th trajectory at time step $t$, respectively; $\hat{R}_t^i$ is the model's predicted return.

The overall training objective of TADT is then formulated as a multi-task loss:
\begin{equation}
\mathcal{L}_{\text{TADT}} = \mathcal{L}_a + \lambda_1 \mathcal{L}_{\text{rank}} + \lambda_2 \mathcal{L}_R 
\end{equation}
, where $\lambda_1$ and $\lambda_2$ are hyperparameters that control the relative weights of the respective loss components.

\subsection{Contrastive State Abstraction}
Due to the extremely large and sparse state space in industrial RS, we introduce a novel CSA module into the TADT model. The CSA module consists of two components: conditioned state quantization and state auxiliary networks.
\subsubsection{Conditioned State Quantization}
Specifically, we first employ a state encoder, like MLP, to project the high-dimensional state embedding $\mathbf{s}_t$ into a low-dimensional latent space and obtain the hidden representation $\mathbf{e}_t$. Inspired by VQ-VAE \cite{van2017neural} and online clustering methods \cite{ma2019learning, xie2016unsupervised}, we propose a Temporal Advantage–conditioned State Vector Quantization (TAC-SVQ) strategy.

Formally, given a codebook $\mathbf{C}$ of size $M$ in the latent space, we define the similarity between the encoder output $e_t$ and each codebook vector $\mathbf{c}_i$, conditioned on the TA score embedding $\mathbf{R}^{\text{TA}}_{t}$, as follows:
\begin{equation}
\begin{aligned}
z(\mathbf{e}_t, \mathbf{c}_i, \mathbf{R}^{\text{TA}}_{t}) &= \alpha \mathbf{c}_i^{T}\mathbf{e}_t + (1 - \alpha) \mathbf{c}_i^{T}\mathbf{R}_{t}^{\text{TA}}\\
p_{\text{sim}}(\mathbf{e}_t, \mathbf{c}_i | \mathbf{R}_{t}^{\text{TA}}) &= \frac{\text{exp}(z(\mathbf{e}_t, \mathbf{c}_i, \mathbf{R}_{t}^{\text{TA}}))}{\sum_{j = 1}^M \text{exp}(z(\mathbf{e}_t, \mathbf{c}_j, \mathbf{R}_{t}^{\text{TA}}))}
\end{aligned}
\label{eq:11}
\end{equation}
where $\alpha \in [0, 1]$ is a hyperparameter that balances the contributions of the current state and the historical TA signal. The initial TA embedding $\mathbf{R}_{0}^{\text{TA}}$ is initialized with $\mathbf{0}$.

Subsequently, the Gumbel-Softmax relaxation \cite{jang2016categorical} is applied to the similarity distribution $\mathbf{p}_{sim}$ to generate a differentiable one-hot assignment vector $\mathbf{z}_t$. The final quantized state embedding $\mathbf{c}_t$ is then obtained via:
\begin{equation}
\mathbf{c}_t = \mathbf{C} \mathbf{z}_t
\end{equation}

To alleviate the issue of codebook collapse \cite{dhariwal2020jukebox, zhang2022deep}, where only a small subset of codebook vectors are frequently used during training, we introduce a regularization loss that encourages diversity in codebook usage. Specifically, we maximize the entropy of the codebook assignment distribution:
\begin{equation}
\mathcal{L}_{\text{reg}} = \sum_{i=1}^N p_{\text{sim}}(\mathbf{e}_t, \cdot | \mathbf{R}_{t}^{\text{TA}}) \log p_{\text{sim}}(\mathbf{e}_t, \cdot | \mathbf{R}_{t}^{\text{TA}})
\end{equation}

\subsubsection{State Auxiliary Networks}
From the perspective of state abstraction theory, it is essential to preserve MDP equivalence between the original state space $\mathcal{S}$ and the compressed latent state space $\mathcal{C}$. To achieve this, we introduce two auxiliary networks: the Reward Prediction (RP) network and the Contrastive Transition Prediction (CTP) network , which not only help satisfy the theoretical constraints but also facilitate learning of expressive and effective state representations.

Formally, the RP network predicts the reward based on $\mathbf{c}_t$ and $a_t$, defined as:
\begin{equation}
\hat{r}_t = \text{MLP}(\mathbf{c}_t, a_t)
\end{equation}

For the CTP network, we adopt a contrastive learning approach inspired by SimCLR \cite{chen2020simple}, aiming to capture transition dynamics in the compressed space. Specifically, we treat the triplet $(\mathbf{c}_t, a_t, \mathbf{c}_{t+1}$) as a positive sample, while treating $(\mathbf{c}_t, a_t, \mathbf{c}'_{t+1}$) as negative samples, where $\mathbf{c}'_{t+1}$ denotes a randomly selected embedding satisfying $c'_{t+1} \neq c_{t+1}$ and $c'_{t+1} \neq c_{t}$. In practice, we sample $K_{\text{neg}}$ negative examples for each positive pair during training.

The predicted logit for the CTP network is then defined as:
\begin{equation}
z'(s_t, a_t, s_{t+1}) = \text{MLP}(\mathbf{c}_t, a_t, \mathbf{c}_{t+1})
\end{equation}

We define the corresponding loss functions for reward prediction and contrastive transition prediction as follows:
\begin{equation}
\begin{aligned}
\mathcal{L}_{r} &= \sum_{i=1}^{N}\sum_{t = 1}^T \|\hat{r}^i_t - r^i_t\|^2 \\
\mathcal{L}_{c} &= - \sum_{i=1}^{N}\sum_{t = 1}^{T-1} \log \frac{\text{exp}\left(z'(s_t, a_t, s_{t+1})\right)}{\sum_{k = 1}^{K_{\text{neg}}} \text{exp}\left(z'(s_t, a_t, s^{(k)}_{t+1})\right)}
\end{aligned}
\end{equation}
where $\mathcal{L}_r$ measures the mean squared error between predicted and ground-truth rewards, and $\mathcal{L}_c$ encourages the model to distinguish between true and false transitions using a contrastive objective.

The overall loss function of the CSA module is then formulated as a weighted combination of these objectives:
\begin{equation}
\mathcal{L}_{\text{CSA}} = \lambda_3\mathcal{L}_{r} + \lambda_4\mathcal{L}_{c} + \lambda_5\mathcal{L}_{\text{reg}}
\end{equation}
where $\lambda_3$, $\lambda_3$ and $\lambda_3$ are hyperparameters that control the relative importance of each component.

Finally, the total loss function of the proposed TADT-CSA model is given by:
\begin{equation}
\mathcal{L} = \mathcal{L}_{\text{TADT}} + \mathcal{L}_{\text{CSA}}
\end{equation}
which combines the main TADT-based policy learning objective with the auxiliary losses from the CSA module.

\subsubsection{Theoretical Analysis}
Here, we present a theoretical analysis to illustrate the rationality and effectiveness of the RP network and CTP network. The proof of the theorem is provided in the Appendix.

% \begin{theorem}
% \label{th:1}
% Let $f_{\theta}$ denote state abstraction function, $g_{\phi}$ the reward prediction network, and $h_{\psi}$ the transition prediction network. Define the maximum reward prediction error as 
% \begin{equation}
% \varepsilon_r = \max_{s, a} \left|r(s, a) - r_g(f_{\theta}(s), a)\right|,
% \end{equation}
% and the maximum transition prediction error as
% \begin{equation}
% \varepsilon_{\mathcal{T}} = \max_{s, a, s'} \left|\mathcal{T}(s, a, s') - \mathcal{T}_{h}(f_{\theta}(s), a, f_{\theta}(s'))\right|.
% \end{equation}
% Let $\Theta = \left\{\theta, \phi, \psi\right\}$. If the $f_{\theta}$ satisfies the following condition:
% \begin{equation}
% \begin{aligned}
% &f_{\theta}(s_1) = f_{\theta}(s_2) \Rightarrow \\
% &\forall a,\quad |r_g(f_{\theta}(s_1), a) - r_g(f_{\theta}(s_2), a)| < \varepsilon, \\
% &\forall s',\quad |\mathcal{T}_h(f_{\theta}(s_1), a, f_{\theta}(s')) - \mathcal{T}_h(f_{\theta}(s_2), a, f_{\theta}(s'))| < \varepsilon,
% \end{aligned}
% \end{equation}
% and suppose that $ \varepsilon_r = \varepsilon_\mathcal{T} = \varepsilon $, then for all states $s$, the value difference between the policy of the abstracted MDP and the optimal policy of the original MDP is bounded by: 
% \begin{equation}
% V_{\pi^{*}}(s) - V_{\pi_{\Theta}}(s) \leq \frac{6\varepsilon + 6\gamma (|S| - 1)\varepsilon}{(1 - \gamma)^3}.    
% \end{equation}
% \end{theorem}

\begin{theorem}
\label{th:1}
Let $f_{\theta}$ denote the state abstraction function mapping states to codebook indices, $g_{\phi}$ the reward prediction network, and $h_{\psi}$ the contrastive transition prediction network, reward function $r(s_t, a_t) \in [0, 1]$. Define the maximum reward prediction error as
\begin{equation}
\varepsilon_r = \max_{s, a} \left|r(s, a) - r_{g_{\phi}}(f_{\theta}(s), a)\right|,
\end{equation}
and the maximum transition prediction error as
\begin{equation}
\varepsilon_{\mathcal{P}} = \max_{s, a, s'} \left|\mathcal{P}(s, a, s') - \mathcal{P}_{h_{\psi}}(f_{\theta}(s), a, f_{\theta}(s'))\right|.
\end{equation}
Let $\Theta = \left\{\theta, \phi, \psi\right\}$. Then for all states $s$, the value difference between the optimal policy and the abstracted policy is bounded by:
\begin{equation}
V_{\pi^{*}}(s) - V_{\pi_{\Theta}}(s) \leq  \frac{2}{(1 - \gamma)^2}\left(\varepsilon_r + \kappa I^{\frac{d+2}{2d}}|\mathcal{C}|^{-\frac{1}{d}} + \frac{\gamma \varepsilon_{\mathcal{P}}|\mathcal{C}|}{1 - \gamma}\right)
\end{equation}
where $\kappa$ is a Lipschitz constant, $I = \int p(\mathbf{e})^{\frac{d}{d+2}} d\mathbf{e}$ is the distribution concentration factor ($0 < I \leq 1$), $d$ is the embedding dimension and $|\mathcal{C}|$ is the codebook size.
\end{theorem}
This theorem guarantees that the proposed CSA module preserves MDP equivalence in the compressed state space, with approximation error independent of original state space size $|\mathcal{S}|$. The error bound depends only on: (1) the fixed codebook size $|\mathcal{C}|$, (2) the prediction errors of the RP and CTP networks, and (3) the concentration of the state embeddings. The property explains the robustness of CSA in industrial recommendation scenarios, where the state space is typically extremely large and dynamically evolving. Furthermore, this theoretical guarantee is particularly well-aligned with the inductive bias of DTs, which implicitly learn policies by reconstructing high-return trajectories conditioned on future RTG signals. By preserving critical reward and transition dynamics through the RP and CTP networks, the abstracted state space retains sufficient MDP structure, thereby enabling DTs to achieve near-optimal decision-making even under compressed representations.

% This theorem guarantees that the proposed CSA module preserves MDP equivalence in the compressed state space with error independent of the original state space size $|\mathcal{S}|$. The bound depends only on: 1) fixed codebook size $|\mathcal{C}|$, 2) prediction errors of RP/CTP networks, and 3) the concentration of state embeddings. This explains CSA's robustness in industrial scenarios where $|\mathcal{S}|$ is often extremely large and grows dynamically. Furthermore, this theoretical guarantee is highly relevant for DT, as they implicitly learn policies by reconstructing optimal behaviors conditioned on future returns. The RP and CTP networks preserve critical reward and transition structures, enabling DT to achieve near-optimal decision-making in the abstracted state space.

\section{Evaluation}
In this section, we evaluate the effectiveness of the proposed TADT-CSA methods on public offline datasets, online simulation environment and online A/B tests.

\subsection{Datasets and Baselines}
We utilize 4 public datasets in the offline evaluation experiments, including KuaiRand-Pure \cite{gao2022kuairand}, MovieLens-20M, Netflix and RetailRocket. The statistics of the above datasets is provided in the Appendix.

The compared methods include traditional offline RL methods, like CQL \cite{kumar2020conservative}, IQL \cite{kostrikov2021offline}; sequential recommendation model, like SASRec \cite{kang2018self} and BERT4Rec \cite{sun2019bert4rec} and DT-based recommendation methods, like DT4Rec \cite{zhao2023user}, CDT4Rec \cite{wang2023causal} and DT4IER \cite{liu2024sequential}.

% 在Offline Evaluation小节前插入表格
\begin{table*}[!ht]
\centering
\footnotesize
\caption{The overall offline performance of all compared methods.}
\label{tab:vertical_comparison}
\begin{tabular}{@{}l l c c c c c c@{}}
\toprule
\textbf{Dataset} & \textbf{Metric} & 
\textbf{SASRec} & \textbf{BEART4Rec} & \textbf{DT4Rec} & \textbf{DT4IER} & \textbf{CDT4Rec} & \textbf{TADT-CSA} \\
\midrule
\multirow{6}{*}{KuaiRand-Pure} 
& Recall@1 & 0.1899 & 0.2104 & 0.1528 & 0.2272 & 0.2402 & \best{0.2704} \\
& Recall@5 & 0.5388 & 0.5313 & 0.4208 & 0.5450 & 0.5753 & \best{0.6136} \\
& Recall@10 & 0.7132 & 0.6908 & 0.5874 & 0.7026 & 0.7359 & \best{0.7679} \\
& NDCG@5 & 0.3686 & 0.3752 & 0.2893 & 0.3908 & 0.4137 & \best{0.4493} \\
& NDCG@10 & 0.4250 & 0.4269 & 0.3432 & 0.4417 & 0.4658 & \best{0.4995}\\
& MRR & 0.3501 & 0.3597 & 0.2870 & 0.3754 & 0.3952 & \best{0.4278} \\
\midrule

\multirow{6}{*}{MovieLens-20M} 
& Recall@1 & 0.1854 & 0.2511 & 0.2603 & 0.2540 & 0.2636 & \best{0.2666} \\
& Recall@5 & 0.5523 & 0.5939 & 0.5941 & 0.5642 & 0.6099 & \best{0.6223} \\
& Recall@10 & 0.7284 & 0.7350 & 0.7389 & 0.6993 & 0.7522 & \best{0.7700} \\
& NDCG@5 & 0.3741 & 0.4307 & 0.4351 & 0.4164 & 0.4441 & \best{0.4517} \\
& NDCG@10 & 0.4280 & 0.4766 & 0.4822 & 0.4602 & 0.4904 & \best{0.4996} \\
& MRR & 0.3514  & 0.4079 & 0.4144 & 0.3984 & 0.4201 & \best{0.4264} \\
\midrule

\multirow{6}{*}{Netflix} 
& Recall@1 & 0.3264 & 0.3567 & 0.3872 & 0.3978 & 0.3959  & \best{0.4007}\\
& Recall@5 & 0.6499 & 0.6406 & 0.6580 & 0.6665 & 0.6854 & \best{0.6878} \\
& Recall@10 & 0.7807 & 0.7685 & 0.7688 & 0.7726 & 0.7996 & \best{0.8033} \\
& NDCG@5 & 0.4969 & 0.4922 & 0.5304 & 0.5405 & \best{0.5514} & 0.5504 \\
& NDCG@10 & 0.5394  & 0.5337 & 0.5663 & 0.5749 & \best{0.5884} & 0.5879 \\
& MRR & 0.4747  & 0.4715 & 0.5135 & 0.5232 & \best{0.5320} & 0.5300 \\
\midrule

\multirow{6}{*}{RetailRocket} 
& Recall@1 & 0.3846 & 0.3806 & 0.3319 & 0.2764 & 0.4088 & \best{0.4430} \\
& Recall@5 & 0.4401 & 0.4695 & 0.4031 & 0.3405 & 0.4758 & \best{0.5157} \\
& Recall@10 & 0.4558 & 0.4932 & 0.4330 & 0.3903 & 0.5057 & \best{0.5427} \\
& NDCG@5 & 0.4152 & 0.4296 & 0.3707 & 0.3096 & 0.4449 & \best{0.4832} \\
& NDCG@10 & 0.4203 & 0.4374 & 0.3804 & 0.3258 & 0.4542 & \best{0.4922} \\
& MRR & 0.4187 & 0.4302 & 0.3766 & 0.3217 & 0.4514 & \best{0.4867} \\
\bottomrule
\end{tabular}
\label{table:2}
\end{table*}

\subsection{Implementation Details}
In the offline experiments, all datasets are preprocessed to construct trajectories of length 30 for each evaluated method. Each state is represented as a 20-dimensional observation vector. For TADT-CSA model, we set the learning rate to 5e-3, the optimizer to Adam, the batch size to 128, the number of training epochs to 50, the codebook size to 64, and the hidden embedding dimension to 64. The implementation of TADT-CSA is based on the PyTorch framework.

\subsection{Offline Evaluation}
For offline experiments, we utilize Recall@K, NDCG@K and MRR as the evaluation metrics to assess prediction accuracy of all compared methods. The detailed explanation of these metrics can be found in Appendix. The results of the offline evaluation are presented in Table \ref{table:2}. Compared to traditional sequential recommendation methods and DT-based approaches, the proposed TADT-CSA method consistently achieves the highest or near-highest scores in Recall, NDCG, and MRR across all four datasets. These results demonstrate the effectiveness and superiority of TADT-CSA over the baseline methods.

\subsection{Online Simulation}
We utilize VirtualTaobao \cite{shi2019virtual} as the online simulation environment to evaluate the performance of all compared methods. First, a DDPG model is trained within this environment and then used as an expert agent to collect trajectories. These trajectories serve to pre-train the offline RL and DT-based methods. Subsequently, all methods are fine-tuned through interaction with the simulation environment. During the fine-tuning stage, we compare the CTR achieved by each method. The results of the online simulation is shown in Fig \ref{fig:2}. This indicates that, compared to existing DT-based and traditional reinforcement learning methods, the proposed TADT-CSA achieves a higher CTR during online fine-tuning, demonstrating its ability to effectively improve policies through interactions with the environment.
\begin{figure}[htbp]
\centering
\includegraphics[width=0.45 \textwidth]{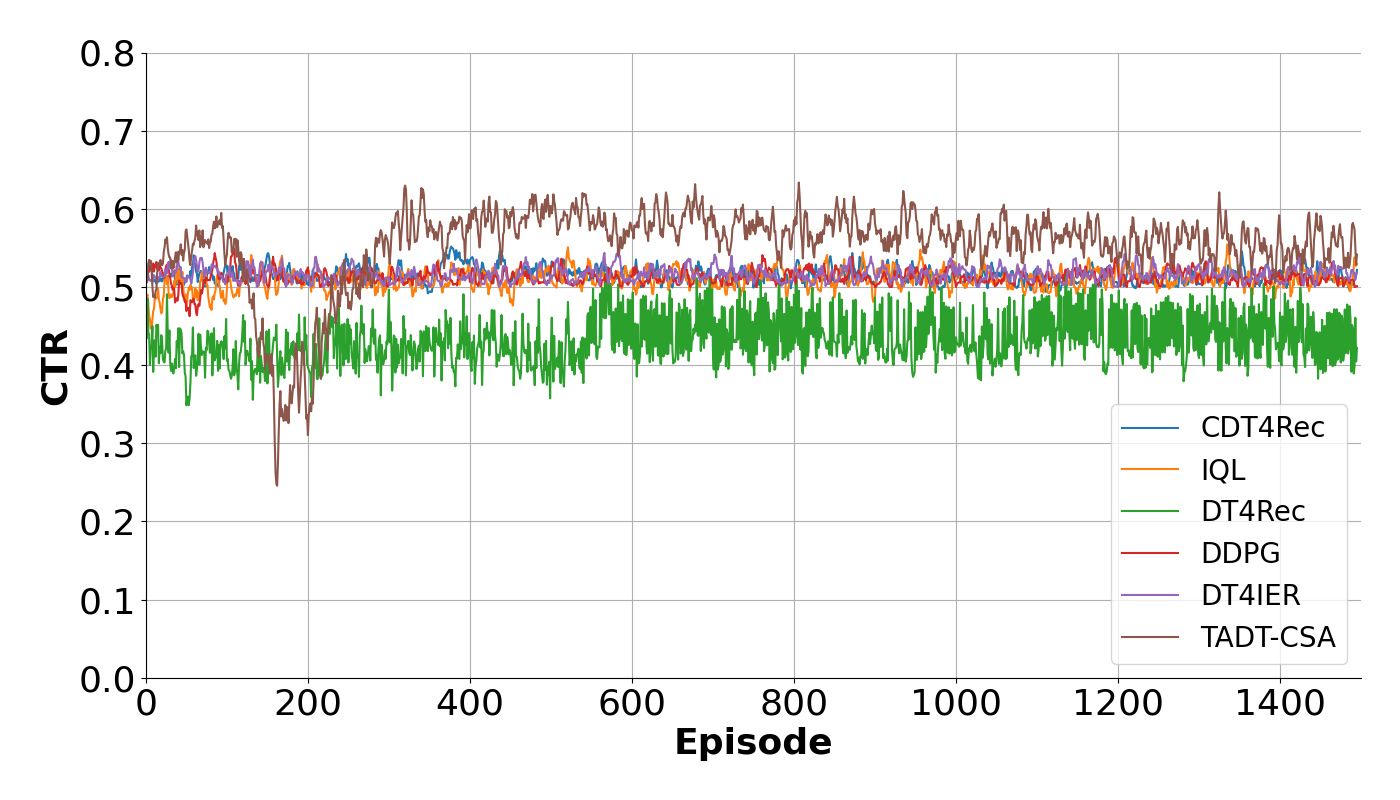}
\caption{Online Simulation result of all compared methods.}
\label{fig:2}
\end{figure}

\subsection{Ablation Study}
To further analyze the superiority of TADT-CSA, we evaluate the offline performance of TADT-CSA and its variants by progressively removing modules, features, or loss functions from the model. The detailed information of all variants is provided in the Appendix. The results of the ablation study are presented in Table \ref{table:3}. TADT-CSA outperforms all its variants in the offline experiments, demonstrating the effectiveness of both the CSA module and the TA signal. Furthermore, we observe that TADT-CSA (w/o TAC \& CTP \& RP) achieves the poorest performance, indicating that the RP and CTP networks play a crucial role in guiding the learning of the codebook. Relying solely on the DT loss for codebook updates may lead to suboptimal state representations.
\begin{table*}[!ht]
\centering
\caption{The offline performance of TADT-CSA and its variants on the KuaiRand-Pure dataset.}
\label{tab:ablation}
\footnotesize
\begin{tabular}{@{}lcccccc@{}}
\toprule
\textbf{Model} & \textbf{Recall@1} & \textbf{Recall@5} & \textbf{Recall@10} & \textbf{MRR} & \textbf{NDCG@5} & \textbf{NDCG@10} \\
\midrule
TADT-CSA & \textbf{0.2704} & 0.6136 & 0.7679 & \textbf{0.4278} & \textbf{0.4493} & \textbf{0.4995} \\
\midrule
TADT-CSA(w/o TAC) & 0.2566 & 0.5940 & 0.7478 & 0.4113 & 0.4316 & 0.4815 \\
TADT-CSA(w/o TAC \& CTP) & 0.2643 & \textbf{0.6155} & \textbf{0.7688} & 0.4237 & 0.4468 & 0.4966 \\
TADT-CSA(w/o TAC \& CTP \& RP) & 0.0893 & 0.2992 & 0.4766 & 0.2061 & 0.1939 & 0.2508 \\
TADT-CSA(w/o CSA) & 0.2578 & 0.6054 & 0.7594 & 0.4163 & 0.4385 & 0.4885 \\
TADT-CSA(w/o CSA \& TA) & 0.2438 & 0.5732 & 0.7313 & 0.3962 & 0.4143 & 0.4656 \\
\bottomrule
\end{tabular}
\label{table:3}
\end{table*}

\subsection{Parameter Sensitivity}
We evaluate the parameter sensitivity of TADT-CSA on the KuaiRand dataset and find that a codebook size of $64$ and $\delta = 0.3$ yield the best performance, as shown in Figure~\ref{fig:3}. Additional parameter sensitivity analyses are provided in the Appendix.
\begin{figure}[htbp]
\centering
\includegraphics[width=0.45 \textwidth]{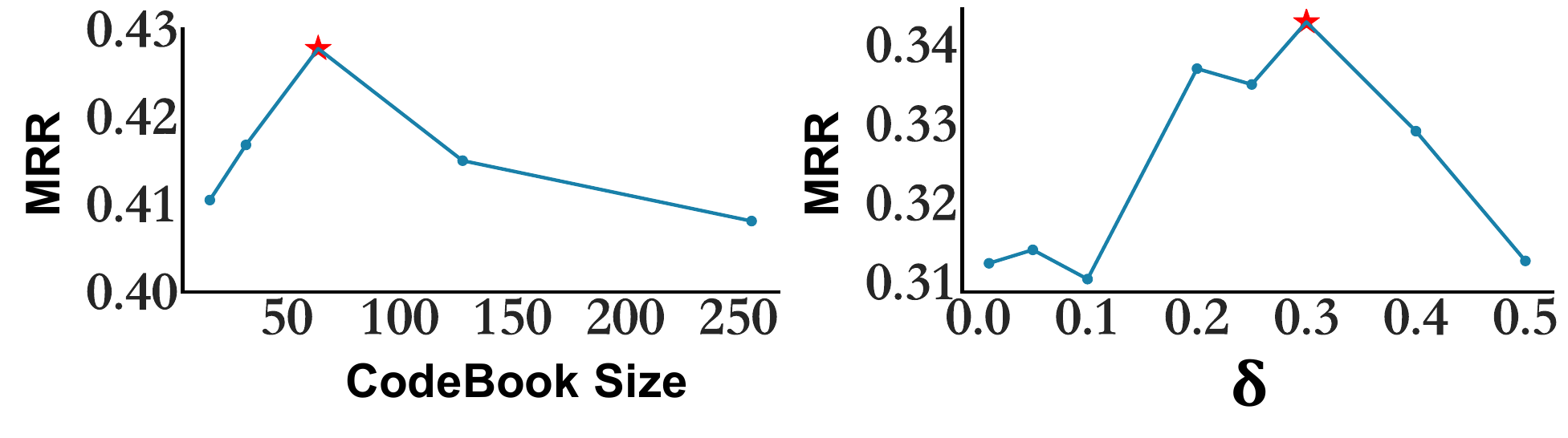}
\caption{Parameter sensitivity evaluation of TADT-CSA.}
\label{fig:3}
\end{figure}

\subsection{Online A/B Experiments}
We deploy TADT-CSA in the ranking stage of the live stream recommendation system on Kwai, a short video platform with over 100 million users. An online A/B test is conducted over a 5-day period in March 2025 to evaluate the performance of TADT-CSA against baseline methods. Specifically, 20\% of Kwai users are randomly selected as the experimental group. The baseline methods include SAC and SAC-CSA, where SAC-CSA is an enhanced version of SAC equipped with the CSA module. In the online experiments, SAC-CSA improves live stream watch time by 10.322\% and average watch time by 12.947\% over SAC. Furthermore, TADT-CSA achieves an additional 2.830\% increase in live stream watch time and a 15.307\% improvement in average watch time compared to SAC-CSA. These results highlight the effectiveness of the CSA module in enhancing representation expressiveness, as well as the capability of the TADT framework in capturing long-term user preferences.

\section{Related Work}
\subsection{Decision Transformer} The Decision Transformer (DT) is a return-conditioned supervised learning(RCSL) method \cite{brandfonbrener2022does}, which has been widely adopted in offline RL and generative recommendation systems. However, DT essentially performs behavior cloning on offline datasets and is not capable of learning optimal policies in the same way as dynamic programming (DP)-based offline RL methods such as CQL \cite{kumar2020conservative} and IQL \cite{kostrikov2021offline}.  Some recent works, including QDT \cite{yamagata2023q} and ACT \cite{gao2024act}, incorporate DP-based offline RL techniques to estimate optimal Q-values or advantage values. These values are then used to relabel or replace the original RTG signals. However, these approaches typically follow a two-stage training paradigm for DT, which can be both complex and time-consuming in practice. Another approach, QT \cite{hu2024q}, introduces a Q-learning loss to complement the behavior cloning objective of DT. However, QT requires training an additional value function separately, increasing model complexity and training overhead.

\subsection{Generative Recommendation} Compared with traditional recommendation methods that rely on user-item scoring and ranking, generative recommendation (GenRec) leverages sequence modeling or generative models—such as large language models (LLMs)—to perform conditional item generation \cite{ji2024genrec}. The Decision Transformer (DT) is a lightweight generative model that has been widely applied in RS, including DT4Rec \cite{zhao2023user}, CDT4Rec \cite{wang2023causal}, MaskRDT \cite{wang2024retentive}, EDT4Rec \cite{chen2024maximum}, and DT4IER \cite{liu2024sequential}. However, these methods often neglect the importance of state representation learning, relying instead on simple embedding layers or encoders to obtain state embeddings. This may lead to suboptimal performance in industrial RS scenarios, where accurate modeling of user states is crucial. 

\section{Conclusion}
In this paper, we propose a novel Temporal Advantage Decision Transformer with Contrastive State Abstraction (TADT-CSA) framework for generative recommendation through long-sequence modeling. The proposed method effectively addresses the challenge that traditional TD error-based RL agents struggle to learn accurate Q-value estimates in highly noisy, stochastic and dynamically fluctuating environments. Furthermore, we introduce the Contrastive State Abstraction (CSA) module into the TADT framework to enable more effective and expressive state representation learning, particularly under the extremely large and sparse state spaces in industrial RS. Compared with existing DT-based approaches, we design a pairwise ranking loss that prevents the model from merely performing behavior cloning and encourages policy improvement based on RTG signals. Extensive experimental results demonstrate the effectiveness and superiority of the TADT-CSA model. Currently, the TADT-CSA framework has been successfully deployed in real-world industrial RS scenarios, where it consistently outperforms existing RL methods based on TD learning, such as SAC \cite{haarnoja2018soft}, in both performance and stability.

% \bigskip
% \noindent Thank you for reading these instructions carefully. We look forward to receiving your electronic files!
\bibliography{aaai2026}

% Reproducibility Checklist
%\clearpage
\section{Reproducibility Checklist}

\begin{enumerate}
    \item This paper:
    \begin{itemize}
     \item Includes a conceptual outline and/or pseudocode description of AI methods introduced (\underline{yes}/partial/no/NA)
     \item Clearly delineates statements that are opinions, hypothesis, and speculation from objective facts and results (\highlight{yes}/no)
     \item Provides well marked pedagogical references for less-familiare readers to gain background necessary to replicate the paper (\highlight{yes}/no)
    \end{itemize}

    \item Does this paper make theoretical contributions? (\highlight{yes}/no) \\
    If yes, please complete the list below:
    \begin{itemize}
        \item All assumptions and restrictions are stated clearly and formally. (\highlight{yes}/partial/no)
        \item All novel claims are stated formally (e.g., in theorem statements). (\highlight{yes}/partial/no)
        \item Proofs of all novel claims are included. (\highlight{yes}/partial/no)
        \item Proof sketches or intuitions are given for complex and/or novel results. (\highlight{yes}/partial/no)
        \item Appropriate citations to theoretical tools used are given. (\highlight{yes}/partial/no)
        \item All theoretical claims are demonstrated empirically to hold. (\highlight{yes}/partial/no/NA)
        \item All experimental code used to eliminate or disprove claims is included. (\highlight{yes}/no/NA)
    \end{itemize}

    \item Does this paper rely on one or more datasets? (\highlight{yes}/no) \\
    If yes, please complete the list below:
    \begin{itemize}
        \item A motivation is given for why the experiments are conducted on the selected datasets (\highlight{yes}/partial/no/NA)
        \item All novel datasets introduced in this paper are included in a data appendix. (\highlight{yes}/partial/no/NA)
        \item All novel datasets introduced in this paper will be made publicly available upon publication of the paper with a license that allows free usage for research purposes. (yes/partial/no/\highlight{NA})
        \item All datasets drawn from the existing literature (potentially including authors’ own previously published work) are accompanied by appropriate citations. (\highlight{yes}/no/NA)
        \item All datasets drawn from the existing literature (potentially including authors’ own previously published work) are publicly available. (\highlight{yes}/partial/no/NA)
        \item All datasets that are not publicly available are described in detail, with explanation why publicly available alternatives are not scientifically satisficing. (yes/partial/no/\highlight{NA})
    \end{itemize}

    \item Does this paper include computational experiments? (\highlight{yes}/no) \\
    If yes, please complete the list below:
    \begin{itemize}
        \item This paper states the number and range of values tried per (hyper-) parameter during development of the paper, along with the criterion used for selecting the final parameter setting. (\highlight{yes}/partial/no/NA)
        \item Any code required for pre-processing data is included in the appendix. (\highlight{yes}/partial/no).
        \item All source code required for conducting and analyzing the experiments is included in a code appendix. (\highlight{yes}/partial/no)
        \item All source code required for conducting and analyzing the experiments will be made publicly available upon publication of the paper with a license that allows free usage for research purposes. (\highlight{yes}/partial/no)
        \item All source code implementing new methods have comments detailing the implementation, with references to the paper where each step comes from (\highlight{yes}/partial/no)
        \item If an algorithm depends on randomness, then the method used for setting seeds is described in a way sufficient to allow replication of results. (yes/partial/no/\highlight{NA})
        \item This paper specifies the computing infrastructure used for running experiments (hardware and software), including GPU/CPU models; amount of memory; operating system; names and versions of relevant software libraries and frameworks. (\highlight{yes}/partial/no)
        \item This paper formally describes evaluation metrics used and explains the motivation for choosing these metrics. (\highlight{yes}/partial/no)
        \item This paper states the number of algorithm runs used to compute each reported result. (\highlight{yes}/no)
        \item Analysis of experiments goes beyond single-dimensional summaries of performance (e.g., average; median) to include measures of variation, confidence, or other distributional information. (yes/\highlight{no})
        \item The significance of any improvement or decrease in performance is judged using appropriate statistical tests (e.g., Wilcoxon signed-rank). (yes/partial/\highlight{no})
        \item This paper lists all final (hyper-)parameters used for each model/algorithm in the paper’s experiments. (\highlight{yes}/partial/no/NA)
    \end{itemize}
\end{enumerate}

\newpage
\onecolumn
\appendix
\section{Appendix}
\subsection{Proofs of Theorem 1}
% \begin{theorem}
% \label{th:1}
% Let $f_{\theta}$ denote the state abstraction function mapping states to codebook indices, $g_{\phi}$ the reward prediction network, and $h_{\psi}$ the contrastive transition prediction network, reward function $r(s_t, a_t) \in [0, 1]$. Define the maximum reward prediction error as
% \begin{equation}
% \varepsilon_r = \max_{s, a} \left|r(s, a) - r_{g_{\phi}}(f_{\theta}(s), a)\right|,
% \end{equation}
% and the maximum transition prediction error as
% \begin{equation}
% \varepsilon_{\mathcal{P}} = \max_{s, a, s'} \left|\mathcal{P}(s, a, s') - \mathcal{P}_{h_{\psi}}(f_{\theta}(s), a, f_{\theta}(s'))\right|.
% \end{equation}
% Let $\Theta = \left\{\theta, \phi, \psi\right\}$. Then for all states $s$, the value difference between the optimal policy and the abstracted policy is bounded by:
% \begin{equation}
% V_{\pi^{*}}(s) - V_{\pi_{\Theta}}(s) \leq  \frac{2}{(1 - \gamma)^2}\left(\varepsilon_r + \kappa I^{\frac{d+2}{2d}}|\mathcal{C}|^{-\frac{1}{d}} + \frac{\gamma \varepsilon_{\mathcal{P}}|\mathcal{C}|}{1 - \gamma}\right)
% \end{equation}
% where $\kappa$ is a Lipschitz constant, $I = \int p(\mathbf{e})^{\frac{d}{d+2}} d\mathbf{e}$ is the distribution concentration factor ($0 < I \leq 1$), $d$ is the embedding dimension and $|\mathcal{C}|$ is the codebook size.
% \end{theorem}

\textbf{\textit{Proof}}. We first formalize the relationship between the original MDP $\mathcal{M} = (\mathcal{S}, \mathcal{A}, \mathcal{P}, \mathcal{R}, \gamma)$ and the abstracted MDP $\mathcal{M}_{\mathcal{C}} = (\mathcal{C}, \mathcal{A}, \mathcal{P}_{\mathcal{C}}, \mathcal{R}_{\mathcal{C}}, \gamma)$ where $|\mathcal{C}|$ is fixed.

\begin{lemma}
Suppose the state embedding distribution p(e) has compact support and satisfies $p(\mathbf{e}) \geq \rho_{\min} >0$ on its support. Define $I=\int p(\mathbf{e})^{\frac{d}{d+2}} d\mathbf{e}$. Then there exists a codebook $C$ with $|C|=N$ such that the covering radius satisfies:
\begin{equation}
\max_{\mathbf{e}} \min_{\mathbf{c}_k \in \mathcal{C}} |\mathbf{e} - \mathbf{c}_k| \leq \kappa' \cdot I^{\frac{d+2}{2d}} \cdot |\mathcal{C}|^{-\frac{1}{d}}
\end{equation}
where $\kappa'$ depends on $d$, $\rho_{\min}$, and the diameter of the embedding space.
\end{lemma}
\begin{proof}
By Zador's asymptotic formula \cite{zador1982asymptotic}, there exists a constant $C_d$ and $N_0$ such that for all $N \geq N_0$, there exists a codebook $\mathcal{C}$ satisfying:
$$
\mathbb{E}[\min_k \|\mathbf{e} - \mathbf{c}_k\|^2] \leq C_d \cdot I^{\frac{d+2}{2d}} \cdot |\mathcal{C}|^{-\frac{1}{d}}.
$$
Under the assumption that $p(\mathbf{e}) \geq \rho_{\min} > 0$ on its compact support, \cite{graf2000foundations} show that this implies the covering radius bound stated in the lemma.
\end{proof}

Now consider the Q-value difference for any $s_t$ with abstract representation $c_t = f_\theta(s_t)$:
\begin{equation}
\begin{aligned}
\Delta Q(c_t, a_t) = \left|Q_{\mathcal{S}}(s_t, a_t) - Q_{\mathcal{C}}(c_t, a_t) \right|
\end{aligned}
\end{equation}

Applying the Bellman equation expansion yields:
\begin{equation}
\begin{aligned}
\Delta Q &\leq \underbrace{|r(s_t,a_t) - \mathcal{R}_{\mathcal{C}}(c_t,a_t)|}_{\text{(a)}} + \gamma \underbrace{\left| \sum_{s_{t+1}} \mathcal{P}(s_{t+1}|s_t,a_t) \max Q_{\mathcal{S}}(s_{t+1}) - \sum_{c_{t+1}} \mathcal{P}_{\mathcal{C}}(c_{t+1}|c_t,a_t) \max Q_{\mathcal{C}}(c_{t+1}) \right|}_{\text{(b)}}
\end{aligned}
\end{equation}

\textbf{Bounding term (a)}:

Let $\hat{r}_t = r_g(f_{\theta}(s), a)$ be the predicted reward value by RP network. Then
\begin{equation}
\begin{aligned}
|r(s_t,a_t) - \mathcal{R}_{\mathcal{C}}(c_t,a_t)| &\leq |r(s_t,a_t) - \hat{r}_t| + |\hat{r}_t - \mathcal{R}_{\mathcal{C}}(c_t,a_t)| \
%&\leq \varepsilon_r + |\hat{r}_t - \mathbb{E}_{s \in c_{_t}}[r(s, a_t)]| &\leq \varepsilon_r + \kappa_r I^{\frac{d+2}{2d}} |\mathcal{C}|^{-\frac{1}{d}}
\end{aligned}
\end{equation}

By the definition of $\varepsilon_r$ in Theorem 1, we obtain:
\begin{equation}
|r(s_t,a_t) - \hat{r}_t| = |r(s_t,a_t) - r_g(f_{\theta}(s), a)| \leq \varepsilon_{r}
\end{equation}
Assuming the reward function is $\kappa_r$-Lipschitz continuous with respect to the embedding space:
\begin{equation}
\left|r(s_i, a_t) - r(s_j, a_t)\right| \leq \kappa_{r}\|\mathbf{e}_i - \mathbf{e}_j\|
\end{equation}
we derive:
\begin{equation}
\begin{aligned}
\left|r(s_i, a_t) - r(s_j, a_t)\right| &\leq \kappa_{r}\left(\|\mathbf{e}_i - \mathbf{c}_k\| + \|\mathbf{e}_j - \mathbf{c}_k\|\right) \\
& \leq 2\kappa_{r}\kappa' I^{\frac{d+2}{2d}} \cdot |\mathcal{C}|^{-\frac{1}{d}}
\end{aligned}
\end{equation}
Setting $\kappa_1 = 2\kappa_r\kappa'$, we obtain:
\begin{equation}
\begin{aligned}
\left|r(s_i, a_t) - r(s_j, a_t)\right| \leq \kappa_1 I^{\frac{d+2}{2d}} |\mathcal{C}|^{-\frac{1}{d}}
\end{aligned}
\end{equation}

Consequently:
\begin{equation}
\label{eq:9}
\begin{aligned}
(a) \leq \varepsilon_r + \kappa_1 I^{\frac{d+2}{2d}} |\mathcal{C}|^{-\frac{1}{d}}
\end{aligned}
\end{equation}

\textbf{Bounding term (b)}:

Applying the triangle inequality, we decompose:

\begin{equation}
\begin{aligned}
\text{(b)} & \underbrace{\leq \left| \sum_{s_{t+1}} \mathcal{P}(s_{t+1}|s_t,a_t) \max Q_{\mathcal{S}}(s_{t+1}) - \sum_{c_{t+1}} \hat{\mathcal{P}}(c_{t+1}|c_t,a_t) \max Q_{\mathcal{C}}(c_{t+1}) \right|}_{(b_1)} \\
&+ \underbrace{\left| \sum_{c_{t+1}} \left(\hat{\mathcal{P}}(c_{t+1}|c_t,a_t) - \mathcal{P}_{\mathcal{C}}(c_{t+1}|c_t,a_t)\right) \max Q_{\mathcal{C}}(c_{t+1}) \right|}_{(b_2)}
\end{aligned}
\end{equation}
where term $b_2$ relates to the CTP network's transition prediction error $\varepsilon_{\mathcal{T}}$. Specifically:

\begin{equation}
\begin{aligned}
(b_2) &\leq \sum_{c_{t+1}}\frac{\varepsilon_{\mathcal{P}}R_{max}}{1 - \gamma} \leq \frac{\varepsilon_{\mathcal{P}}|\mathcal{C}|}{1 - \gamma}
\end{aligned}
\end{equation}

For term $(b_1)$, we assume:
\begin{equation}
\hat{\mathcal{P}}(c_{t+1} | c_t, a_t) \approx \mathcal{P}(c_{t+1} | c_t, a_t) = \sum_{s' \in f^{-1}(c_{t+1})} \mathcal{P}(s' | s_t, a_t)
\end{equation}

Assuming the Q-function is $\kappa_{Q}$-Lipschitz continuous with respect to the embedding space:
\begin{equation}
\left|Q(s_i, a_t) - Q(s_j, a_t)\right| \leq \kappa_{Q}\|\mathbf{e}_i - \mathbf{e}_j\|
\end{equation}
we obtain:
\begin{equation}
\begin{aligned}
&\left| \sum_{s_{t+1}} \mathcal{P} \max Q_{\mathcal{S}} - \sum_{c_{t+1}} \hat{\mathcal{P}} \max Q_{\mathcal{C}} \right| \\
&\approx \left| \sum_{c_{t+1}} \sum_{s_{t+1} \in f^{-1}(c_{t+1})} \mathcal{P}(s_{t+1}|s_t,a_t) \left( \max Q_{\mathcal{S}}(s_{t+1}) - \max Q_{\mathcal{C}}(c_{t+1}) \right) \right| \\
&\leq \sum_{c_{t+1}} \sum_{s_{t+1} \in f^{-1}(c_{t+1})} \mathcal{P}(s_{t+1}|s_t,a_t) \left| \max Q_{\mathcal{S}}(s_{t+1}) - \max Q_{\mathcal{C}}(c_{t+1}) \right| \\
&\leq \sum_{c_{t+1}} \sum_{s_{t+1} \in f^{-1}(c_{t+1})} \mathcal{P}(s_{t+1}|s_t,a_t) \cdot \kappa_Q \cdot \| \mathbf{e}_{s_{t+1}} - \mathbf{c}_{c_{t+1}} \| \\
&\leq \kappa_Q \cdot \left( \kappa' I^{-\frac{d+2}{2d}} |\mathcal{C}|^{-\frac{1}{d}} \right)
\end{aligned}
\end{equation}

Defining $\kappa_2 = \kappa_{Q}\kappa'$, we establish:
\begin{equation}
\label{eq:14}
(b) \leq \frac{\varepsilon_{\mathcal{P}}|\mathcal{C}|}{1 - \gamma} + \kappa_2 I^{\frac{d+2}{2d}} |\mathcal{C}|^{-\frac{1}{d}}
\end{equation}

Combining the bounds from Equations \ref{eq:9} and \ref{eq:14} through recursive application of the Bellman equation:
\begin{equation}
\begin{aligned}
\Delta Q &\leq \varepsilon_r + \kappa_1 I^{\frac{d+2}{2d}} |\mathcal{C}|^{-\frac{1}{d}} + \gamma\left(\frac{\varepsilon_{\mathcal{P}}|\mathcal{C}|}{1 - \gamma} + \kappa_2 I^{\frac{d+2}{2d}} |\mathcal{C}|^{-\frac{1}{d}}\right) \\
&\leq \varepsilon_r + (\kappa_1 + \gamma\kappa_2) I^{\frac{d+2}{2d}}|\mathcal{C}|^{-\frac{1}{d}} + \frac{\gamma \varepsilon_{\mathcal{P}}|\mathcal{C}|}{1 - \gamma}
\end{aligned}
\end{equation}

Setting $\kappa = \kappa_1 + \gamma\kappa_2$, we arrive at:

\begin{equation}
\label{eq:17}
\Delta Q \leq \varepsilon_r + \kappa I^{\frac{d+2}{2d}}|\mathcal{C}|^{-\frac{1}{d}} + \frac{\gamma \varepsilon_{\mathcal{P}}|\mathcal{C}|}{1 - \gamma}
\end{equation}

We now restate a fundamental result from \cite{abel2016near}:
\begin{lemma}
\label{lm:2}
Let $f$ be a state abstraction function, if $f_{\varepsilon}(s_1) = f_{\varepsilon}(s_2) \rightarrow \forall_a \left|Q(s_1, a) - Q(s_2, a)\right| < \varepsilon$, then $\forall_s V_{\pi^*}(s) - V_{\pi_{\Theta}}(f(s)) \leq \frac{2\varepsilon}{(1 - \gamma)^2}$.
\end{lemma}

Applying Equation \ref{eq:17} in conjunction with Lemma \ref{lm:2}, we derive the final performance bound:
\begin{equation}
\begin{aligned}
V_{\pi^*}(s) - V_{\pi_{\Theta}}(s) &\leq \frac{2}{(1 - \gamma)^2}\left(\varepsilon_r + \kappa I^{\frac{d+2}{2d}}|\mathcal{C}|^{-\frac{1}{d}} + \frac{\gamma \varepsilon_{\mathcal{P}}|\mathcal{C}|}{1 - \gamma}\right)
\end{aligned}
\end{equation}

\subsection{Dataset Description}
We utilize 4 public datasets in the offline evaluation experiments, including KuaiRand-Pure, MovieLens-20M, Netflix and RetailRocket. The statistics of the above datasets are provided as follows:

\begin{table}[H]
\centering
\begin{tabular}{cccc}
\hline
\rule{0pt}{2.4ex}Methods & \#Users & \#Items & \#Interactions \\
\hline
\rule{0pt}{2.4ex}KuaiRand-Pure & 27,285 & 7,551 & 1,436,609 \\
MovieLens-20M   & 138,493 & 27,278 & 20,000,263 \\
Netflix  & 480,189 & 17,770 & 100,498,277 \\RetailRocket & 1,407,580 & 235060 & 2,756,101 \\
\hline
\end{tabular}
\caption{Dataset Statistics in the offline experiments.}
\label{table:1}
\end{table}

\subsection{Evaluation metrics}
To comprehensively evaluate the performance of recommendation systems, we employ the following industry-standard metrics across all experiments:

\begin{itemize}
    \item \textbf{Recall@K}: Measures the proportion of relevant items found in the top-K recommendations:$$\text{Recall@}K = \frac{|\{\text{relevant items}\} \cap \{\text{top-K recommendations}\}|}{|\{\text{relevant items}\}|}$$We report results for $ K = \{1, 5, 10\} $ to capture performance at different recommendation list lengths. \item \textbf{NDCG@K}: Normalized Discounted Cumulative Gain accounts for ranking quality by assigning higher weights to top positions: $$\text{NDCG@}K = \frac{\text{DCG@}K}{\text{IDCG@}K}, \quad \text{DCG@}K = \sum_{i=1}^{K} \frac{2^{rel_i} - 1}{\log_2(i+1)} $$ where $ rel_i $ is the relevance score of item at position $i$, and IDCG is the ideal DCG. Results are reported for $K = \{5, 10\}$.
    \item \textbf{MRR}: Mean Reciprocal Rank evaluates the ranking of the first relevant item:
$$\text{MRR} = \frac{1}{|Q|} \sum_{i=1}^{|Q|} \frac{1}{\text{rank}_i}$$ where $\text{rank}_i$ is the position of the first relevant item for query $i$.
\end{itemize}
These metrics provide complementary perspectives: Recall captures recommendation coverage, NDCG evaluates ranking quality, and MRR focuses on the position of the first relevant recommendation. We omit NDCG@1 as it's highly correlated with Recall@1 and provides limited additional insight.

\subsection{TADT-CSA variants}
All the TADT-CSA variants is shown as follows:
\begin{itemize}
\item  \textbf{TADT-CSA (w/o TAC)}: removes the TA condition from the similarity calculation in Eq. 11.
\item \textbf{TADT-CSA (w/o TAC \& CTP)}: removes both the TAC in the similarity calculation and the downstream CTP network, relying solely on the RP network to guide codebook updates.
\item \textbf{TADT-CSA (w/o TAC \& CTP \& RP)}: removes the TA condition and all downstream task networks, using only the DT loss to guide codebook updates.
\item \textbf{TADT-CSA (w/o CSA)}: removes the entire CSA module along with the quantile ranking loss, since the quantile ranking loss will use the codebook information. 
\item \textbf{TADT-CSA (w/o CSA \& TA)}: removes the CSA module, quantile ranking loss, and the TA signal in the DT, which degrades to the original DT framework.
\end{itemize}
\subsection{Parameter Sensitivity Figures}
Other parameter sensitivity evaluation experiments are shown in Fig \ref{fig:4}, including the balanced factor $\alpha$ in the codebook similarity calculation, and the quantile threshold $\beta$ in the quantile pairwise ranking loss. We run TADT-CSA with different parameter values of $\alpha$ and $\beta$ on the KuaiRand-Pure dataset. It indicates that the evaluation metric MRR are not sensitive with both parameters. 
\begin{figure*}[htbp]
\centering
\includegraphics[width=0.8\textwidth]{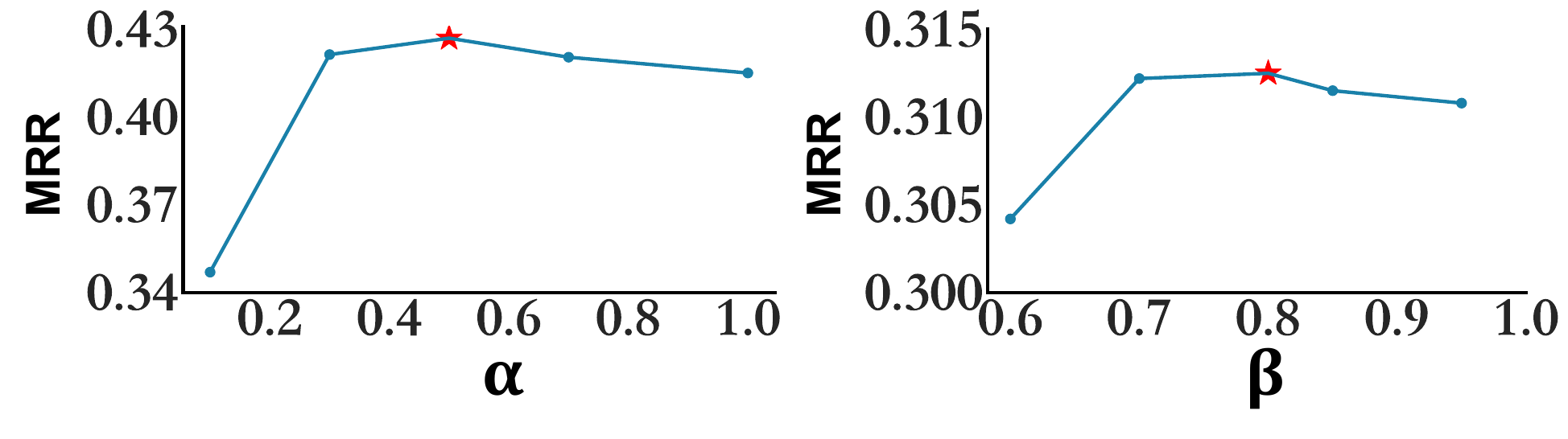}
\caption{Parameter sensitivity evaluation on $\alpha$ and $\beta$.}
\label{fig:4}
\end{figure*}

\subsection{Computational Cost Analysis}
We compared the computational costs of different methods on the KuaiRand-Pure training dataset, as shown in Table \ref{table:4}. The results indicate that DT-based methods are more computationally efficient than LLM-based methods, highlighting the lightweight nature of DT models in terms of shorter running time, lower GPU memory consumption, and faster convergence. Specifically, the running time and GPU memory usage of TADT-CSA are comparable to those of existing DT-based methods, demonstrating that the inclusion of the CSA module does not increase model complexity or computational overhead.

\begin{table}[H]
\centering
\begin{tabular}{cccc}
\hline
\rule{0pt}{2.4ex}Models & \#Running Time & \#GPU Memory & \#Convergence Steps\\
\hline
\rule{0pt}{2.4ex}DT4Rec & 4.5h & 8GB & 14,280 \\
DT4IER   & 3.9h & 8GB & 12,380 \\
CDT4Rec  & 3.4h & 8GB & 9,600 \\
LLM & 18.6h & 24GB & 307,200 \\
TADT-CSA & 4.0h & 12GB & 8,600 \\
\hline
\end{tabular}
\caption{Computational cost of different methods.}
\label{table:4}
\end{table}

\begin{figure*}[htbp]
\centering
\includegraphics[width=1.0\textwidth]{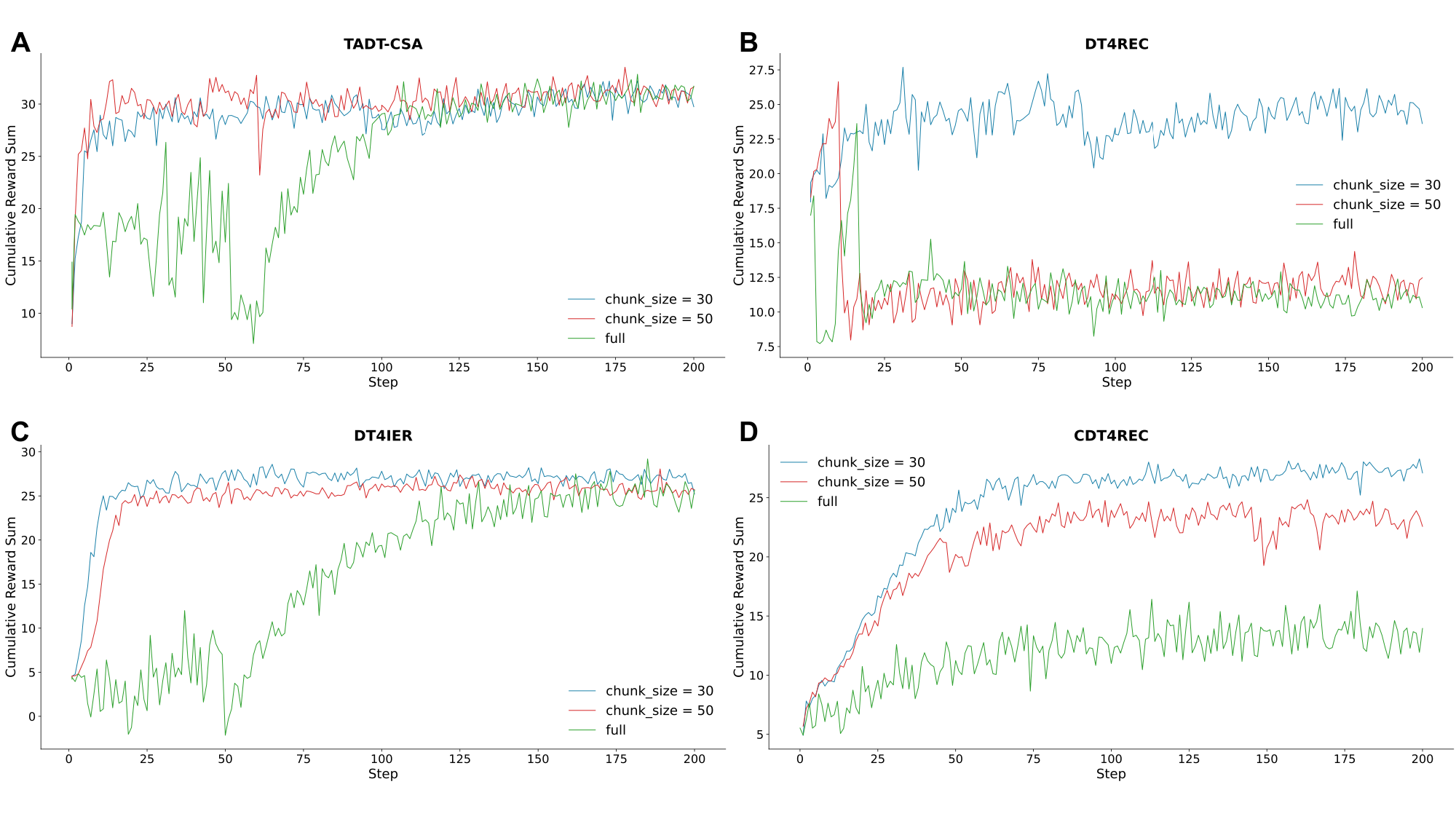}
\caption{Trajectory stitching evaluation of different DT-based recommendation methods.}
\label{fig:5}
\end{figure*}

\subsection{Trajectory Stitching Evaluation}
Since Decision Transformer (DT) models often suffer from the trajectory stitching problem, as discussed in prior works such as ACT \cite{gao2024act}, we conduct experiments on the KuaiRand-Pure dataset to evaluate how different chunk lengths impact model performance. Specifically, we divide full-length trajectories (length = 200) into smaller chunks of fixed lengths (e.g., 30, 50) and compare the cumulative reward achieved by TADT-CSA against other DT-based baselines, including DT4Rec, CDT4Rec and DT4IER.

As shown in Figure \ref{fig:5}, TADT-CSA consistently outperforms all baselines across different chunk sizes, and maintains strong performance even when using the full trajectories. In contrast, other DT-based methods exhibit noticeably degraded performance when using full-length trajectories, despite the fact that full sequences are typically expected to yield richer information and better results. We hypothesize that this counterintuitive behavior stems from the fact that excessively long sequences (e.g., 200 steps) pose significant optimization and modeling challenges for standard DTs, such as over-smoothing or failure to attend to earlier tokens effectively. This makes them more vulnerable to trajectory stitching issues in the full-length setting.

Moreover, DT4Rec and CDT4Rec perform relatively well on shorter chunks (e.g., 30), but their performance drops and convergence slows as chunk size increases. DT4IER, on the other hand, demonstrates stable convergence across all chunk sizes (30, 50, and full), with its cumulative reward reaching similar levels in each case. However, it still underperforms compared to TADT-CSA, indicating that despite its robustness, DT4IER remains slightly less effective in capturing long-term value signals. 

These findings demonstrate the implicit trajectory stitching capability of TADT-CSA and its superior ability to model long-horizon dependencies in sequential recommendation scenarios.

\end{document}